\documentclass[letterpaper, 10 pt, conference]{ieeeconf}
\newif\ifnotes\notestrue 

\usepackage[english]{babel}
\usepackage{makeidx}
\usepackage{amsmath,amssymb,amsfonts}
\usepackage{latexsym}
\usepackage{graphicx}
\usepackage{hyperref}
\usepackage{paralist}
\usepackage{xcolor}
\usepackage{tikz}
\usepackage{setspace}
\usepackage{algorithm}
\usepackage[noend]{algpseudocode}
\usepackage[framemethod=tikz]{mdframed}
\usepackage{xspace}
\usepackage{pgfplots}
\pgfplotsset{compat=1.5}

\newtheorem{theorem}{Theorem}

\newtheorem{lemma}[theorem]{Lemma}

\newtheorem{definition}[theorem]{Definition}

\newtheorem{claim}[theorem]{\bf Claim}

\newtheorem{observation}[theorem]{\bf Observation}
\newtheorem{example}[theorem]{Example}

\newenvironment{proofof}[1]{\begin{trivlist} \item {\bf Proof
#1:~~}}
  {\qed\end{trivlist}}
\renewenvironment{proofof}[1]{\par\medskip\noindent{\bf Proof of #1: \ }}{\hfill$\Box$\par\medskip}
\newcommand{\namedref}[2]{\hyperref[#2]{#1~\ref*{#2}}}
\newcommand{\thmlab}[1]{\label{thm:#1}}
\newcommand{\thmref}[1]{\namedref{Theorem}{thm:#1}}
\newcommand{\obslab}[1]{\label{obs:#1}}
\newcommand{\obsref}[1]{\namedref{Observation}{obs:#1}}
\newcommand{\lemmlab}[1]{\label{lemm:#1}}
\newcommand{\lemmref}[1]{\namedref{Lemma}{lemm:#1}}
\newcommand{\claimlab}[1]{\label{claim:#1}}
\newcommand{\claimref}[1]{\namedref{Claim}{claim:#1}}

\newcommand{\seclab}[1]{\label{sec:#1}}

\newcommand{\figlab}[1]{\label{fig:#1}}
\newcommand{\figref}[1]{\namedref{Figure}{fig:#1}}

\newcommand{\COMMENTED}[1]{{}}

\newcommand{\EEx}[1]{\ensuremath{\mathbf{E}\Big[#1\Big]}}
\newcommand{\eps}{\epsilon}

\newcommand{\ed}{\ensuremath{\mathsf{ed}}}
\newcommand{\HAM}{\ensuremath{\mathsf{HAM}}}

\newcommand{\mha}{\ensuremath{\mathsf{ModifiedHirschberg}}}

\newcommand{\mdef}[1]{{\ensuremath{#1}}\xspace}  

\newcommand{\superscript}[1]{\ensuremath{^{\mbox{\tiny{\textit{#1}}}}}\xspace}
\def \th {\superscript{th}}     

\newcommand{\flr}[1]{\mdef{\left\lfloor#1\right\rfloor}}              
\newcommand{\ceil}[1]{\mdef{\left\lceil#1\right\rceil}}               
\renewcommand{\O}[1]{\ensuremath{\mathcal{O}\left(#1\right)}}						

\newcommand{\ignore}[1]{}

\ifnotes
\newcommand{\elena}[1]{\textcolor{red}{{\bf (Elena:} {#1}{\bf ) }} \marginpar{\tiny\bf
             \begin{minipage}[t]{0.5in}
               \raggedright E:
                \end{minipage}}}
\newcommand{\erfan}[1]{\textcolor{blue}{{\bf (Erfan:} {#1}{\bf ) }} \marginpar{\tiny\bf
             \begin{minipage}[t]{0.5in}
               \raggedright E:
                \end{minipage}}}
\newcommand{\samson}[1]{\textcolor{green}{{\bf (Samson:} {#1}{\bf ) }} \marginpar{\tiny\bf
             \begin{minipage}[t]{0.5in}
               \raggedright S:
            \end{minipage}}}            							
\else
\newcommand{\elena}[1]{}
\newcommand{\erfan}[1]{}
\newcommand{\samson}[1]{}
\fi

\definecolor{mahogany}{rgb}{0.75, 0.25, 0.0}
\definecolor{darkblue}{rgb}{0.0, 0.0, 0.55}
\definecolor{darkpastelgreen}{rgb}{0.01, 0.75, 0.24}
\definecolor{darkgreen}{rgb}{0.0, 0.2, 0.13}
\definecolor{darkgoldenrod}{rgb}{0.72, 0.53, 0.04}
\definecolor{darkred}{rgb}{0.55, 0.0, 0.0}
\hypersetup{
     colorlinks   = true,
     citecolor    = mahogany,
		 linkcolor		= olive
}

\title{Longest Alignment with Edits in Data Streams}

\author{Elena Grigorescu$^{*}$ \and Erfan {Sadeqi Azer}$^{\dagger}$ \and Samson Zhou $^{\ddagger}$
\thanks{$^{*}$Department of Computer Science, Purdue University, West Lafayette, IN.
        {\tt\small elena-g@purdue.edu}. Research supported in part by NSF CCF-1649515.}%
\thanks{$^{\dagger}$School of Informatics and Computing, Indiana University, Bloomington, IN.
        {\tt\small esadeqia@indiana.edu}}%
\thanks{$^{\ddagger}$Department of Computer Science, Purdue University, West Lafayette, IN.
        {\tt\small samsonzhou@gmail.com}. Research supported in part by NSF CCF-1649515.}%
}

\IEEEoverridecommandlockouts

\begin{document}

\maketitle

\begin{abstract}
Analyzing patterns in data streams generated by network traffic, sensor networks, or satellite feeds is a challenge for systems in which the available storage is limited. In addition, real data is noisy, which makes designing data stream algorithms even more challenging.

Motivated by such challenges, we study algorithms for detecting the similarity of two data streams that can be read in sync. 
Two strings $S, T\in \Sigma^n$ form a $d$-near-alignment if the distance between them in some given metric is at most $d$. 
We study the problem of identifying a longest substring of $S$ and $T$ that forms a  {\em $d$-near-alignment}  under the {\em edit} distance, in  the {\em simultaneous streaming model}. 
In this model,  symbols of strings $S$ and $T$ are streamed at the same time, and the amount of available processing space is sublinear in the length of the strings.

We give several algorithms, including an exact one-pass algorithm that uses $\O{d^2+d\log n}$ bits of space. 
We couple these results with comparable lower bounds. 


\end{abstract}
\section{Introduction}

A data stream is a massive sequence of elements (network packets, database transactions, sensor network reads, or parts of nucleic acids) that requires further processing, while it is too large to be stored entirely. The area of streaming algorithms, initiated in \cite{Alon1996}, is now a core subject in computer science, focusing on re-designing classical algorithms to the setting where the amount of available working space is only sublinear in the size of the data.
Furthermore, the area has connections to other modern topics, including sketching, compressed sensing and communication complexity (for comprehensive surveys, see e.g., \cite{Chakrabarti2015b,Muthukrishnan2004,Gilbert2010}).

In this work we are concerned with approximately measuring the similarity between two data streams, by finding a largest `near-alignment'.  
Two strings $S, T \in \Sigma^n$ form a {\em $d$-near-alignment} if the distance between them in some given metric is at most $d$. 
In this paper we consider the  {\em edit} distance (or {\em Levenshtein} distance), which is the minimum number of insertions, deletions, or substitutions needed to obtain one string from the other. 

We study the {\em $d$-Substring Alignment} problem of finding the longest $d$-near-alignment in the edit distance, consisting of substrings  in $S$ and $T$ of the form $(S[i,j], T[i,j]))$,  
when the symbols of  $S$ and $T$ are streamed in sync\footnote{In this paper, all the techniques and results are presented assuming the input is in binary bits. However, all the results can be adapted for non-binary settings.}.

The following definition formally defines $\ell_{max}$, the quantity that is studied in this paper.
\begin{definition}
The length of the longest {\em $d$-near-alignment} between two strings $S$ and $T$, with length $n$, is 
\[\ell_{max}=\max_{1\leq i\leq j\leq n}\{j-i+1\,|\,ed(S[i,j],T[i,j])\leq d\},\]
where $ed(S[i,j],T[i,j])$ denotes the minimum number of insertions, deletions, or substitutions needed to obtain $T[i,j]$ from $S[i,j]$.
\end{definition}

\begin{example}
Let $S=``1234yyyyyy123456789xxxxx"$ and $T=``1234xxxxxx123467890yyyyy"$, and $d=2$. The longest {\em $d$-near-alignment} between $S$ and $T$ is ``123456789'' from $S$ and ``123467890'' from $T$. This implies that $\ell_{max}=9$.
\end{example}

 
Specifically, in the {\em simultaneous streaming model}, the symbols at index $i$ of two strings $S$ and $T$ arrive at the same time, and  the pair $(S[i], T[i])$ arrives right before the pair $(S[i+1], T[i+1])$. 
In the streaming model, the algorithm can only use a small amount of space, ideally sublinear in the length of the input. 
The input may be revealed in one pass or multiple passes, and the goal is to obtain a solution to an optimization problem. One pass algorithms have a wider range of applications. Though, some applications might allow two or more passes over input.

\subsection*{Our results}
We obtain several algorithms and lower bounds for the $d$-Substring Alignment problem in the simultaneous streaming model, as detailed next. We will use $\ell_{max}$ to denote the length of a longest $d$-near-alignment, in the edit distance.

As a warm-up, we start with a multiplicative and an additive approximation. 
 
\newcommand{\thmmultiply}{There exists a one-pass simultaneous streaming algorithm that provides a $(1+\eps)$-approximation to  $\ell_{max}$, using $\O{\frac{d\log^2 n}{\eps\log(1+\eps)}}$ bits of space.}
\begin{theorem} \thmlab{thm:multiply}
\thmmultiply
\end{theorem}

\newcommand{\thmadd}{There exists a one-pass simultaneous streaming algorithm that provides a $d$-near-alignment of length at least $\ell_{max}-E$ using $\O{\left(\frac{n}{E}\right)d\log n}$ bits of space. }
\begin{theorem} \thmlab{thm:add}
\thmadd
\end{theorem}

Our main result is a one-pass, {\em exact} algorithm that outputs a maximum-length $d$-near-alignment using $\O{d^2+d\log n}$ bits of space. Hence, the multiplicative bound from \thmref{thm:multiply} achieves space savings guarantees if the sequence of edits does not need to be printed and $d=\omega(\log^2 n)$. The additive space bound from \thmref{thm:add} achieves better upper-bounds guarantees if we afford $E=\omega\left(\frac{n\log n}{d}\right)$.


\newcommand{\thmmain}{There exists a deterministic one-pass algorithm that outputs $\ell_{max}$, along with the necessary edit operations,  using $\O{d^2+d\log n}$ bits of space.}
\begin{theorem} \thmlab{thm:main}
\thmmain
\end{theorem}

We remark that our algorithms can be extended to the more general case where the substrings of $S$ and $T$ need not begin at the same index. Given the promise that a longest alignment of the two strings begins within $\delta$ indices of each other, one may run $\delta$ instances of our algorithms in parallel, thus incurring an extra factor of $\delta$ in the space complexity.

In terms of lower bounds,  if the edits to obtain the longest $d$-near alignment are output, then we trivially must use $\Omega(d\log n)$ bits of space. A straightforward argument shows that this lower bound holds even if the algorithm is not required to output the positions of the mismatched indices.


\newcommand{\thmlbtwo}{For $\eps<1$ and $E\in R^+$, any deterministic algorithm that computes a $(1+\eps)$-multiplicative, or an $E$-additive approximation of $\ell_{max}$ requires $\Omega(d\log n)$ bits of processing space.}
\begin{theorem}\thmlab{thm:lbtwo}
\thmlbtwo
\end{theorem}


We also give a lower bound for the $d$-Substring Alignment problem in the streaming model where the string $S$ appears before the string $T$ (rather than in sync).

\newcommand{\thmlbone}{For $7<d=o(\sqrt{n})$, any randomized $(1+\eps)$-approximation streaming algorithm computing $\ell_{max}$ with success probability at least $1-1/n$, requires $\Omega(d\log n)$ bits of space.}
\begin{theorem}\thmlab{thm:lbone}
\thmlbone
\end{theorem}

Finally,  we observe that our algorithms can be modified to recognize {\em complementary $d$-near-alignments}, which are objects  relevant to computational biology  arising in pairings of DNA or RNA sequences:
\begin{definition}
Let $f\,:\,\sum\rightarrow\sum$ be a pairing of symbols in the alphabet. A string $S\in\sum^n$ is a complementary alignment if $S[x]=f(T[x])$ for all $1\le x\le n$.
\end{definition}
Indeed, for each character $T[x]$ in $T$, one can feed $f(T[x])$ instead of $T[x]$ to our algorithm in order to find a complementary alignment between $S$ and $T$.

\subsection*{Motivation and related work}

The $d$-Substring Alignment problem is a  restricted variant of the classic \emph{Longest Common Substring} problem, in which the goal is to find a  longest substring common to the given strings $S$ and $T$. It is also related to the {\em Longest Common Subsequence} problem, in which the goal is to find the longest common subsequence  of $S$ and $T$. The offline solutions to these problems involve either suffix trees or dynamic programming \cite{Weiner73, Hui92}. Some of these problems and related string alignment problems have been recently studied in the streaming model (e.g., \cite{Liben-NowellVZ06, SunW07, BerenbrinkEMS14, KociumakaSV14, GawrychowskiMSU16}).

Real data is often subject to errors, and hence algorithms that account for ``near''-alignments, rather than just alignments, are important for processing data. The mismatches leading to near-alignments are most relevant to metrics such as Hamming and edit distance. While the Hamming distance only accounts for substitutions, the edit distance accounts for insertions and deletions, in addition to substitution. Therefore, it is often the case that the study of alignment problems in the edit distance is more challenging than in the Hamming distance. 

Alignment problems have sustained interest in the computer science community over many decades (see e.g., book  \cite{ApostolicoG97}).  The edit metric has been recently well-studied in the streaming model, e.g.,  \cite{AndoniGMP13, BackursI15, ChakrabortyGK16a, ChakrabortyGK16, BelazzouguiZ16}), and ``mismatches'' in the Hamming metric have been investigated in  \cite{PoratP09, LeimeisterM14, AluruAT15, FlouriGKU15, CliffordFPSS16, Starikovskaya16, GrigorescuSZ17, ErgunGSZ17}.

\subsection*{Preliminaries and Overview}
We denote the set $\{1, 2, \ldots, n\}$ by $[n]$.  
We assume that two input streams are strings of length $n$ over a finite alphabet $\Sigma$. 
Given a string $S[1,\dots,n]$, we denote its length by $|S|$, its $i\th$ character by $S[i]$ or $S_i$, and the substring between locations $i$ and $j$ (inclusive) by $S[i,j]$.  

The edit or Levenshtein distance between $S$ and $T$, denoted $\ed(S,T)$, is the minimum number of insertions, deletions, or substitutions needed to obtain one string from the other. 
We say $S[i,j]$ and $T[i,j]$ is a $d$-\emph{near-alignment} if $\ed(S[i,j],T[i,j])\le d$. 
A related metric which we use in proving lower bounds is the Hamming distance. 
The Hamming distance between $S$ and $T$, denoted $\HAM(S,T)$ is the number of indices whose symbols do not match: $\HAM(S,T)=\Big|\{ i\mid S[i]\ne T[i]\}\Big|$.

{\bf The approximation algorithms from \thmref{thm:multiply} and \thmref{thm:add}:}
We define a sequence of {\em checkpoints}, such that at each checkpoint we initiate a sketch of the following characters in each of the two streams, $S$ and $T$, so that we can compare the alignments. 
The checkpoints for the one-pass multiplicative algorithm in \thmref{thm:multiply} are dynamically created and maintained to guarantee the $(1+\eps)$-approximation, as in \figref{fig:checkpoints}, while the checkpoints for the one-pass additive algorithm are predefined.

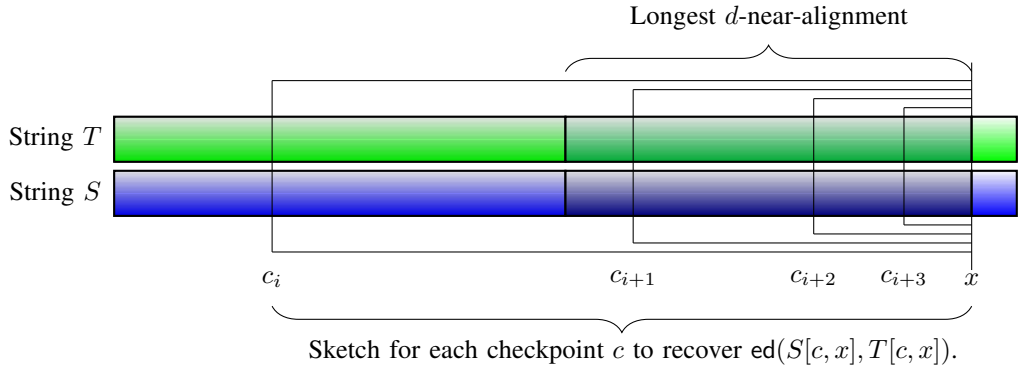
\begin{figure*}[htb]
\centering
\begin{tikzpicture}[scale=0.6]
\filldraw[thick, top color=white,bottom color=blue] (0cm,0cm) rectangle+(10cm,1cm);
\filldraw[thick, top color=white,bottom color=green] (0cm,1.2cm) rectangle+(10cm,1cm);
\filldraw[thick, top color=white,bottom color=darkblue] (10cm,0cm) rectangle+(9cm,1cm);
\filldraw[thick, top color=white,bottom color=darkpastelgreen] (10cm,1.2cm) rectangle+(9cm,1cm);
\filldraw[thick, top color=white,bottom color=blue] (19cm,0cm) rectangle+(1cm,1cm);
\filldraw[thick, top color=white,bottom color=green] (19cm,1.2cm) rectangle+(1cm,1cm);

\draw (0cm,1.2cm) -- (20cm,1.2cm);
\draw (0cm,2.2cm) -- (20cm,2.2cm);
\draw (0cm,1.2cm) -- (0cm,2.2cm);
\draw (20cm,1.2cm) -- (20cm,2.2cm);

\draw (19cm,-1.2cm) -- (19cm, 3.4cm);

\draw (17.5cm,-0.2cm) -- (17.5cm, 2.4cm);
\draw (17.5cm,-0.2cm) -- (19cm, -0.2cm);
\draw (17.5cm,2.4cm) -- (19cm, 2.4cm);

\draw (15.5cm,-0.4cm) -- (15.5cm, 2.6cm);
\draw (15.5cm,-0.4cm) -- (19cm, -0.4cm);
\draw (15.5cm,2.6cm) -- (19cm, 2.6cm);

\draw (11.5cm,-0.6cm) -- (11.5cm, 2.8cm);
\draw (11.5cm,-0.6cm) -- (19cm, -0.6cm);
\draw (11.5cm,2.8cm) -- (19cm, 2.8cm);

\draw (3.5cm,-0.8cm) -- (3.5cm, 3cm);
\draw (3.5cm,-0.8cm) -- (19cm, -0.8cm);
\draw (3.5cm,3cm) -- (19cm, 3cm);

\node at (19cm, -1.4cm){$x$};
\node at (17.5cm, -1.4cm){$c_{i+3}$};
\node at (15.5cm, -1.4cm){$c_{i+2}$};
\node at (11.5cm, -1.4cm){$c_{i+1}$};
\node at (3.5cm, -1.4cm){$c_i$};
\node at (-1.3cm, 0.5cm){String $S$};
\node at (-1.3cm, 1.7cm){String $T$};

\draw [decorate,decoration={brace,mirror,amplitude=10pt}](19cm,3.2cm) -- (10cm,3.2cm);
\node at (14.5cm, 4.4cm){Longest $d$-near-alignment};
\draw [decorate,decoration={brace,mirror,amplitude=10pt}](3.5cm,-2cm) -- (19cm,-2cm);
\node at (11.5cm, -3cm){Sketch for each checkpoint $c$ to recover $\ed(S[c,x],T[c,x])$.};
\end{tikzpicture}
\caption{Checkpoints spaced to guarantee $(1+\eps)$-approximation.}\figlab{fig:checkpoints}
\end{figure*}

For each checkpoint $c$, we create a sketch of $S[c,x]$, using the data structure from \cite{BelazzouguiZ16}, which uses $\O{d\log n}$ bits of space. This sketch is indeed relevant to the  simultaneous streaming model.

\newcommand{\thmbz}{There exists a data structure in the simultaneous streaming model that computes the edit distance using $\O{d\log n}$ bits of space and $\O{n+d^2}$ processing time. Furthermore, this data structure can be augmented to recover the necessary edit operations, using $\O{d^2\log n}$ bits of space.}
\begin{theorem}\thmlab{thm:bz}\cite{BelazzouguiZ16}
\thmbz
\end{theorem}

Upon reading $S[x]$ and $T[x]$, for each checkpoint $c$ we compare the sketches of $S[c,x]$ and $T[c,x]$ using \cite{BelazzouguiZ16} (\thmref{thm:bz}).
If the edit distance is greater than $d$, we discard the sketches. 
Otherwise, we compare $x-c+1$ to our estimate of the length of the longest $d$-near-alignment and proceed with the stream. 
We give further details about how the structure updates from $S[c,x]$ to $S[c,x+1]$ shortly.

To obtain the additive approximation guaranteed by the one-pass algorithm in \thmref{thm:add}, we modify our checkpoints, so that they appear in every $E$ positions. 
Hence, the longest $d$-near-alignment contains a checkpoint within $E$ positions of the its first position, and the algorithm will recover a $d$-near-alignment with length at least $\ell_{max}-E$.

For the sake of completeness, we now briefly describe the {\it Belazzougui-Zhang (BZ) Sketch} \cite{BelazzouguiZ16} (\thmref{thm:bz}) mentioned above.
Recall that the edit distance between two strings in the classic offline model can be solved through dynamic programming, such as in the Needleman-Wunsch and Wagner-Fischer algorithms \cite{Vintsiuk68, NeedlemanW70, Sankoff72, Sellers74, WagnerF74}.
The dynamic programming solution involves creating an alignment matrix, namely an $n\times n$ matrix whose $ij$\th entry contains the value of $\ed(S[1,i],T[1,j])$, called the \emph{score} of that entry. 
The BZ data structure outputs $\ed(S[1,x],T[1,x])$ by keeping a sketch of the alignment matrix, size $\O{d\log n}$, as well as some additional information to mimic the recursive solution in the offline model. 
Upon seeing $S[x+1]$ and $T[x+1]$, it updates the sketch by performing the same recursion as the classic offline dynamic programming solution.

Specifically, the BZ sketch notes that for aligned strings with edit distance at most $d$, at most $2d+1$ diagonals need to be considered, as in \figref{fig:matrix}. 
The sketch maintains a key invariant: the scores of any two adjacent diagonals can differ by at most 1.

\begin{figure*}[htb]
\centering
\begin{tikzpicture}[scale=1]

\draw [->] (0cm,0cm) -- (4.4cm,0cm);
\draw [->] (0cm,0cm) -- (0cm,4.4cm);
\draw (4cm,0cm) -- (4cm,4cm);
\draw (0cm,4cm) -- (4cm,4cm);

\draw [->] (0cm,0cm) -- (4cm,4cm);
\draw [->] (0cm,0.4cm) -- (3.6cm,4cm);
\draw [->] (0cm,0.8cm) -- (3.2cm,4cm);
\draw [->] (0cm,1.2cm) -- (2.8cm,4cm);
\draw [->] (0.4cm,0cm) -- (4cm,3.6cm);
\draw [->] (0.8cm,0cm) -- (4cm,3.2cm);
\draw [->] (1.2cm,0cm) -- (4cm,2.8cm);

\draw (0cm,2cm) -- (4cm,2cm);
\draw (2cm,0cm) -- (2cm,4cm);

\node at (0.1cm, -0.6cm){$S[1]$};
\node at (-0.6cm, 0.1cm){$T[1]$};
\node at (2cm, -0.6cm){$S[d]$};
\node at (-0.6cm, 2cm){$T[d]$};
\node at (4cm, -0.6cm){$S[2d]$};
\node at (-0.6cm, 4cm){$T[2d]$};
\end{tikzpicture}
\caption{The BZ sketch mimics dynamic programming (essentially Figure 4 in \cite{BelazzouguiZ16})}\figlab{fig:matrix}
\end{figure*}
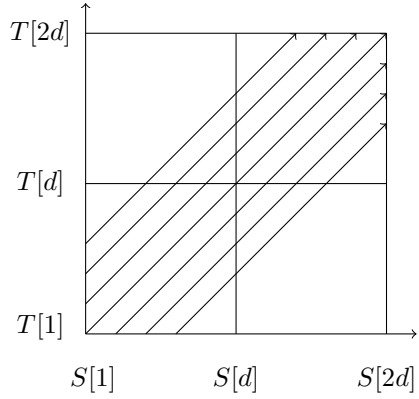

The algorithm maintains a suffix tree that allows computation of the longest common prefix of suffixes of $S$ and $T$.  
Thus, the algorithm updates the score for each diagonal by mimicking dynamic programming, based on the scores of the adjacent diagonals, information from the suffix tree, as well as additional information on the location of the most recent edit operation in each diagonal.

{\bf Our one-pass exact algorithm in \thmref{thm:main}} bypasses the use of the BZ sketch from \cite{BelazzouguiZ16}, to obtain improved space guarantees. Indeed, while one may use the BZ sketch here too for $\O{d^2\log n}$ bits of space, our algorithm uses $\O{d^2+d\log n}$ bits of space.

Our approach is based on a couple of important observations.
First, no character in $S$ may be aligned to a character in $T$ that is at least $d+1$ positions away.
Thus, if there exist $d+1$ consecutive positions in $S$ that are aligned to $d+1$ consecutive positions in $T$, then we only need to keep the locations of the $d$ most recent edit operations before this region. 
Therefore, any $(d+1)^2$ sequence of consecutive positions either contains such a region (where $d+1$ consecutive positions in $S$ are aligned to $d+1$ consecutive positions in $T$), or requires at least $d$ edit operations in order to be aligned.
The algorithm maintains a sliding window of size $(d+1)^2$ as well as the locations of the $d$ most recent edit operations, allowing recovery of the longest $d$-near-alignment.

However, straightforward recovery of the edit operations in the sliding window using a BZ sketch takes $\O{d^2\log n}$ bits. 
To improve on this space complexity, we modify the classical Hirschberg's algorithm \cite{Hirschberg75}. 
Recall that  Hirschberg algorithm is a dynamic programming algorithm that finds the optimal sequence alignment between two strings of length $n$ using $\O{n\log n}$ bits of space.
It uses divide-and-conquer to split each string into two substrings, and recursively compares the optimal sequence alignment between the corresponding substrings. 
We use the algorithm here on the sliding window of length $\O{d^2}$, but because we are only interested in finding alignments with edit distance at most $d$, we can allow the Hirschberg algorithm to throw away any alignments with edit distance more than $d$. 
This modification, detailed in the proof of \thmref{thm:mha}, produces an algorithm that uses $\O{d^2+d\log n}$ bits of space. 

{\bf The lower bounds} Finally, to  show the lower bound from \thmref{thm:lbone} we construct distributions for which any deterministic algorithm fails with significant probability unless given a certain amount of space, and then apply Yao's principle. 
We first reduce the problem of approximating the longest $d$-near-alignment under the edit distance to the problem of approximating longest $d$-near-alignment under the Hamming distance. 
We then reduce the problem to exactly identifying whether two strings have Hamming distance at most $d$. 
We construct hard distributions, and show via counting arguments that deterministic algorithms using  ``low'' amounts of space fail  on inputs from these distributions. 
\section{The Multiplicative Approximation Algorithm}
In this section, we prove \thmref{thm:multiply}, giving a $\O{\frac{d\log^2 n}{\eps\log(1+\eps)}}$ space, one-pass streaming algorithm with multiplicative approximation $(1+\eps)$ to the length of the longest $d$-near-alignment under the edit distance.
Furthermore, the algorithm uses $\O{\frac{(nd+d^3)\log^2 n}{\eps\log(1+\eps)}}$ update time per arriving symbol.  

Prior to the stream, we initialize the list of checkpoints $\mathcal{C}$ to be the empty set, and $\tilde{\ell}$ (the current estimate of the length of the longest $d$-near-alignment) to be zero. 
We dynamically create and maintain the checkpoints to guarantee the $(1+\eps)$-approximation. 
At each checkpoint, we initiate a BZ sketch for each of the two streams, $S$ and $T$, so that we can compare the alignments. 
We also set $c_{start}$, the beginning position of the returned $d$-near-alignment, to be zero. 
The algorithm in full appears below.
\begin{mdframed}
Maintenance:
\begin{enumerate}
\item
Read $S[x], T[x]$.
\item
For each checkpoint $c\in\mathcal{C}$, update the sketches of $\ed(S[c,x-1],T[c,x-1])$ to $\ed(S[c,x],T[c,x])$ respectively.
\item
For all $k\ge k_0$:
\begin{enumerate}
\item
If $x$ is a multiple of $\flr{\alpha(1+\alpha)^{k-2}}$, where $\alpha=\sqrt{1+\eps}-1$. then add the checkpoint $c=x$ to $\mathcal{C}$. Set $\textsf{level}(c)=k$.
\item
If there exists a checkpoint $c$ with $\textsf{level}(c)=k$ and $c<x-2(1+\alpha)^k$, then delete $c$ from $\mathcal{C}$.
\end{enumerate}
\item
For every checkpoint $c\in\mathcal{C}$ such that $x-c+1>\tilde{\ell}$, check if $S[c,x]$ and $T[c,x]$ are $d$-near-alignments. 
If $S[c,x]$ and $T[c,x]$ are $d$-near-alignments, then set $c_{start}=c$, $\tilde{\ell}=x-c+1$.
\item
If $x=n$, then report $c_{start}$ and $\tilde{\ell}$.
\end{enumerate}
\end{mdframed}

Because the checkpoints are spaced as the same as \cite{BerenbrinkEMS14}, then the following properties hold:

\begin{observation}\cite{BerenbrinkEMS14}\obslab{obs:checkpoints}
At reading $S[x]$, for all $k\ge k_0=\ceil{\frac{\log\left(\frac{(1+\alpha)^2}{\alpha}\right)}{\log(1+\alpha)}}$, let $C_{x,k}=\{c\in\mathcal{C}\,|\,\textsf{level}(c)=k\}$.
\begin{enumerate}
\item
$C_{x,k}\subseteq[x-2(1+\alpha)^k,x]$.
\item
The distance between two consecutive checkpoints of $C_{x,k}$ is $\flr{\alpha(1+\alpha)^{k-2}}$.
\item
$|C_{x,k}|=\ceil{\frac{2(1+\alpha)^k}{\flr{\alpha(1+\alpha)^{k-2}}}}$.
\item 
At any point in the algorithm, the number of checkpoints is $\O{\frac{\log n}{\eps\log(1+\eps)}}$.
\end{enumerate}
\end{observation}

\begin{proofof}{\thmref{thm:multiply}}
Let $\ell_{max}$ be the length of the longest $d$-near-alignment, between indices $i_{max}$ and $j_{max}$. 
Let $k$ be the largest integer so that $2(1+\alpha)^{k-1}<\ell_{max}$, where $\alpha=\sqrt{1+\eps}-1$. 
Therefore, $j_{max}-2(1+\alpha)^{k-1}>i_{max}$. 

By \obsref{obs:checkpoints}, there exists a checkpoint in the interval $[j_{max}-2(1+\alpha)^{k-1},j_{max}]$. 
Moreover, \obsref{obs:checkpoints} also implies that consecutive checkpoints of level $k-1$ are separated by distance $\flr{\alpha(1+\alpha)^{k-2}}$. 
Thus, there exists a checkpoint $c$ in the interval $\left[j_{max}-2(1+\alpha)^{k-1},j_{max}-2(1+\alpha)^{k-1}+\alpha(1+\alpha)^{k-3}\right]$. 
Hence, the output $\tilde{\ell}$ of the algorithm is at least $2(1+\alpha)^{k-1}-\alpha(1+\alpha)^{k-3}$.
Thus, the output of the algorithm satisfies the approximation guarantee
\[\frac{\ell_{max}}{\tilde{\ell}}\le\frac{2(1+\alpha)^k}{2(1+\alpha)^{k-1}-2\alpha(1+\alpha)^{k-3}}\]
\[=\frac{(1+\alpha)^3}{(1+\alpha)^2-\alpha}\le(1+\alpha)^2=1+\eps.\]

Since there are at most $\frac{\log n}{\eps\log(1+\eps)}$ checkpoints at any point, and each sketch $S[c_i,x]$ uses $\O{d\log n}$ space, then the total space used is $\O{\frac{d\log^2 n}{\eps\log(1+\eps)}}$. 
As each sketch requires $\O{n+d^2}$ update time, the total update time is $\O{\frac{(nd+d^3)\log^2 n}{\eps\log(1+\eps)}}$.  
\end{proofof}
\section{The Additive Approximation Algorithm}
\seclab{sec:add}
In this section, we prove \thmref{thm:add}, giving a $\O{\left(\frac{n}{E}\right)d\log n}$ space, one-pass streaming algorithm returning the length of the longest $d$-near-alignment under the edit distance, with additive error at most $E$.
Unlike the previous algorithm that uses a series of dynamic checkpoints, this algorithm creates and maintains a checkpoint for every multiple of $E$. 
Again, the checkpoints ``sandwich'' the longest $d$-near-alignment within an additive window of size $E$. 
Before the stream begins, we initialize $\tilde{\ell}$, the current estimate of the length of the longest $d$-near-alignment to be zero and $c_{start}$, the beginning position of the returned $d$-near-alignment, to be zero. 
Then upon seeing characters $S[x]$ and $T[x]$ in the stream:
\begin{mdframed}
Maintenance:
\begin{enumerate}
\item
Read $S[x], T[x]$. 
\item
For each checkpoint $c$, update the sketches of $\ed(S[c,x-1],T[c,x-1])$ to

 $\ed(S[c,x],T[c,x])$,  respectively.
\item
If $x$ is a multiple of $E$, then add the checkpoint $c=x$ to $\mathcal{C}$.
\item
For every checkpoint $c\in\mathcal{C}$ such that $x-c+1>\tilde{\ell}$, we check if $S[c,x]$ and $T[c,x]$ are $d$-near-alignments. 
If $S[c,x]$ and $T[c,x]$ are $d$-near-alignments, then set $c_{start}=c$, $\tilde{\ell}=x-c+1$.
\item
If $x=n$, then report $c_{start}$ and $\tilde{\ell}$.
\end{enumerate}
\end{mdframed}
We now show correctness of \thmref{thm:add}.
\begin{proofof}{\thmref{thm:add}}
Let $\ell_{max}$ be the length of the longest $d$-near-alignment, between indices $i_{max}$ and $j_{max}$. 
If $j_{max}-i_{max}\le E$, then the result holds trivially. 
Otherwise, $i_{max}+E<j_{max}$ and there exists a checkpoint in the interval $[i_{max},i_{max}+E]$, since the checkpoints are spaced distance $E$ apart. 
From the correctness of the BZ sketch, the checkpoint will find a $d$-near-alignment, and so the output of the algorithm will be at least $j_{max}-(i_{max}+E)+1\ge\ell_{max}-E$. 
Thus, the correctness of the algorithm follows.

Since we keep a sketch for each multiple of $E$, there are $\frac{n}{E}$ checkpoints. 
Each sketch is of size $\O{d\log n}$ bits, so the total space used is $\O{\left(\frac{n}{E}\right)d\log n}$. 
\end{proofof}
\section{The Longest $d$-Near-Alignment Algorithm}
In this section, we present a one-pass streaming algorithm that returns the longest $d$-near-alignment with space $\O{d^2+d\log n}$ bits, thus proving \thmref{thm:main}. We emphasize that the algorithm is deterministic.

The idea is to distinguish between the following two cases: either all edit operations corresponding to the longest $d$-near-alignment are close to each other, or there is at least one pair of consecutive edit operations that are at least $d$ indices apart. 
We show that if the second case holds, so that there is at least one pair of consecutive edit operations that are at least $d$ indices apart, it suffices to keep the locations of the $d$ most recent edit operations before this region. 
To this end, our algorithm stores the information of the optimal alignment for the region of the input before a long-enough gap of identical substrings, along with all the characters in a sliding window of a length at most $(d+1)^2$.	

Consider a sliding window beginning at some position $b$ and ending with the most recent position, $x$. 
We enforce an invariant for this window: the edit operations corresponding to the optimal alignment within this window are always at most $d$ positions apart from each other.
We ultimately show in \lemmref{lem:window} that this property ensures the sliding window has size at most $(d+1)^2$.

However, na{\"i}vely recovering the edit operations in the sliding window takes $\O{d^2\log n}$ bits. 
Thus, we detail modifications of the classical Hirschberg algorithm, called procedure $\mha$, in \thmref{thm:mha} to guarantee $\O{d^2+d\log n}$ space. 
While the classical Hirschberg algorithm is a dynamic programming algorithm that finds the optimal sequence alignment between two strings, the promise that our alignment contains at most $d$ edits allow us to greatly narrow the search space.

Let $\mathcal{A}$ denote the set of the most recent $d$ edit operations corresponding to the optimal alignment between for $S[0,b]$ and $T[0,b]$.

In summary, the algorithm stores the following data:
\begin{itemize}
\item The indices $b$ and $x$.
\item The characters of $S[b,x]$ and $T[b,x]$.
\item The set of at most $d$ edit operations $\mathcal{A}$, in a queue data structure.
\item The information of the longest $d$-near-alignment found so far, namely:
\begin{itemize}
\item $i_s$, $j_s$: the two ends of the $d$-near-alignment, so that $\ell=j_s-i_s+1$ is the length of the longest $d$-near-alignment
\item $\mathcal{L}$: the set of edit operations.  
\end{itemize}
\end{itemize}
\begin{mdframed}
Maintenance:
\begin{enumerate}
\item
Read $S[x], T[x]$. 
\item
Construct the optimal alignment between $S[b,x]$ and $T[b,x]$ using $\mha$ algorithm. If there exist $d+1$ consecutive positions in $S$ that are aligned to $d+1$ consecutive positions in $T$, i.e., $S[i_1,j_1]=T[i_2,j_2]$ with $j_1-i_1=j_2-i_2>d$, then 
\begin{enumerate}
\item Append the at most $d$ latest edit operations corresponding to indices before $i_1$ and $i_2$ to $\mathcal{A}$ from the optimal alignment between $S[b,x]$ and  $T[b,x]$.
\item Remove earlier operations from $\mathcal{A}$, until $|\mathcal{A}|\leq d$. 
\item Update $b=\min\{j_1,j_2\}$.
\end{enumerate}
\item Identify whether $\ed(S[b,x],T[b,x])$ is greater than $d$ using $\mha$ algorithm. 
\item If $\ed(S[b,x],T[b,x])>d$, then define $c$ to be the smallest index in $[b,x]$ such that $\ed(S[c,x],T[c,x])\leq d$. Note that $c$ is also computable with $\mha$ algorithm. 
\item Else if $\ed(S[b,x],T[b,x])\le d$, let $f=\ed(S[b,x],T[b,x])$, and define $c$ be the index of $(d-f)$\th operation from the end in $\mathcal{A}$. 
\item In either case, check if  $x-c+1>\ell$, then update $i_s,j_s,\ell,\mathcal{L}$ accordingly.
\item
If $x=n$, then report $\mathcal{L}$ and $\ell$.
\end{enumerate}
\end{mdframed}

\newcommand{\thmmha}{
Given two strings $S$ and $T$ of length $m$ and a parameter $d$, there exists an algorithm that either states that $\ed(S,T)>d$ or recovers the locations of the edit operations if $\ed(S,T)\le d$, using $\O{m+d\log m}$ space and $\O{md\log m}$ time.
}
\begin{theorem}[$\mha$]\thmlab{thm:mha}
\thmmha
\end{theorem}
\begin{proof}
The classic Hirschberg's algorithm \cite{Hirschberg75, kleinberg2006algorithm} returns the locations of the optimal edit operations between $S$ and $T$ in $\O{m\log m}$ space.
However, if we do not care about the locations of the edit operations for $\ed(S,T)>d$, then we can optimize the space down to $\O{m+d\log m}$ bits using ideas from \cite{Ukkonen85}. 

In the classic Hirschberg algorithm, the edit distance is computed for multiple alignments. 
Specifically, the entry $ij$ in the dynamic programming lookup table contains the edit distance between the substrings $S[1,i]$ and $T[1,j]$.
However, if $|j-i|>d$, then the edit distance between $S[1,i]$ and $T[1,j]$ is greater than $d$. 
Therefore, at each level of Hirschberg's algorithm, we only keep $2d-1$ cells around the main diagonal (a similar idea is used for the BZ sketch in \figref{fig:matrix}). 
If $\ed(S,T)>d$, then some optimal edit operation will appear outside of the cells that we keep. 
Thus, the algorithm recognizes that it cannot recover the optimal operations, and instead declares $\ed(S,T)>d$.  
Hence, if $\ed(S,T)\le d$, the algorithm will return the locations of the optimal edit operations, whereas if $\ed(S,T)>d$, the algorithm outputs $\ed(S,T)>d$. 
Since each cell contains a score using $\log m$ bits, the total space used is $\O{m+d\log m}$. 

Recall that Hirschberg's algorithm uses a divide-and-conquer approach, splitting the dynamic programming table into two subproblems, roughly of equal size, say $q$ and $m-q$, where $\left|q-\frac{m}{2}\right|\le d$.
At each level, with input size $m'$, the algorithm takes $\O{m'd}$ time. 
Hence, the algorithm satisfies the recursion $T(m)=\O{md}+T(m-q)+T(q)$ so that the overall running time is $\O{md\log m}$.
\end{proof}

\newcommand{\lemwin}{Let $x,y\in \Sigma^h$ be two strings of length $h$.
Also let $\mathcal{A}$ be the set of all edit operations corresponding to the optimal alignment between $x$ and $y$. 
If $e$ is the maximum distance between two operations among all consecutive operations in $\mathcal{A}$, then we have: $h\leq (|\mathcal{A}|+1)(e+1)$. }
\begin{lemma} \lemmlab{lem:window}
\lemwin
\end{lemma}
\begin{proof}
Suppose, by way of contradiction, $h>(|\mathcal{A}|+1)(e+1)$. 
Since $e$ is the maximum distance between the locations of two operations among all consecutive operations in $\mathcal{A}$, then any group of $e+1$ consecutive characters contains an edit operations. 
But there are at least $|\mathcal{A}|+1$ disjoint groups of $e+1$ consecutive characters, so there are at least $|\mathcal{A}|+1$ edit operations. 
This contradicts the definition that $\mathcal{A}$ is the set of all edit operations.
\end{proof}
We now show the correctness of \thmref{thm:main}.
\begin{proofof}{\thmref{thm:main}}
Let $i$ and $j$ be the two endpoints of the longest $d$-near-alignment. 
Also, let $\mathcal{L}$ be the set of edit operations corresponding to the optimal alignment between $S[i,j]$ and $T[i,j]$. 
There are two cases for this alignment. 
Either no two consecutive operations in $\mathcal{L}$ have distance farther than $d+1$ or there exist $d+1$ consecutive positions in $S$ which are aligned to $d+1$ consecutive positions in $T$.

In the first case, the correctness follows from the correctness of \thmref{thm:bz} and \lemmref{lem:window}. 
In this case, $S[i,j]$ and $T[i,j]$ will be covered by the sliding window after reading $S[j]$ and $T[j]$. 
This means that in line (3), $x=j$ and the algorithm will assign $c=i$. 
Thus, the algorithm will report the correct $d$-near-alignment.

Suppose the second case occurs. 
So, there exist $d+1$ consecutive positions in $S[i,j]$ are aligned to $d+1$ consecutive positions in $T[i,j]$, i.e., $S[i_1,j_1]=T[i_2,j_2]$ with $j_1-i_1=j_2-i_2>d$. 
We claim that no character before $i_1$ ($i_2$, respectively) in $S$ ($T$, respectively) could be aligned to any character after $j_2$ ($j_1$, respectively) in an optimal alignment between $S[i,j]$ and $T[i,j]$, as in \figref{fig:align}.

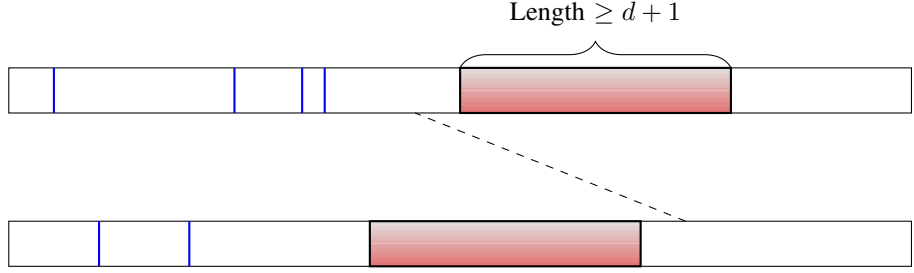
\begin{figure*}[htb]
\centering
\begin{tikzpicture}[scale=0.6]
\draw (0cm,0cm) rectangle +(20cm,1cm);
\draw (0cm,3.4cm) rectangle +(20cm,1cm);
\filldraw[thick, top color=white,bottom color=red!50!] (8cm,0cm) rectangle+(6cm,1cm);
\filldraw[thick, top color=white,bottom color=red!50!] (10cm,3.4cm) rectangle+(6cm,1cm);
\draw [decorate,decoration={brace,mirror,amplitude=10pt}](16cm,4.4cm) -- (10cm,4.4cm);
\node at (13cm, 5.6cm){Length $\ge d+1$};

\draw[thick, color=blue] (4cm,0cm) -- (4cm, 1cm);
\draw[thick, color=blue] (2cm,0cm) -- (2cm, 1cm);
\draw[thick, color=blue] (5cm,3.4cm) -- (5cm, 4.4cm);
\draw[thick, color=blue] (1cm,3.4cm) -- (1cm, 4.4cm);
\draw[thick, color=blue] (6.5cm,3.4cm) -- (6.5cm, 4.4cm);
\draw[thick, color=blue] (7cm,3.4cm) -- (7cm, 4.4cm);


\draw [dashed] (15cm,1cm) -- (9cm,3.4cm);

%
%
\end{tikzpicture}
\caption{If there exists some alignment in which the red regions are aligned, then nothing before the region can be aligned to anything after the region (the dashed alignment can never exist). Thus, it suffices to keep the locations of the $d$ most recent edit operations before the region (for example, the blue lines)}\figlab{fig:align}
\end{figure*}

Otherwise, more than $d$ insertions or deletions are required.

Therefore, the algorithm will recover the alignment between $S[i,i_1]$ and $T[i,i_2]$ from what it has already stored in $\mathcal{A}$. 
In addition, the alignment between $S[i_1,j]$ and $T[i_2,j]$ is constructed at line (2) and these two alignments are combined in line (4).

The space needed to store $\mathcal{A}$ and $\mathcal{L}$ is $\O{d\log n}$ as there are at most $2d$ operations in each data structure. 
Taking $|\mathcal{A}|\le d$ and $e=(d+1)$ in \lemmref{lem:window} implies that the sliding window $[b,x]$ is $\O{d^2}$ bits long. 
Taking $m=(d+1)^2$ in \thmref{thm:mha} implies that recovery of the edit operations can be done using $\O{d^2+d\log n}$ space.
Hence, the overall memory of the algorithm is $\O{d^2+d\log n}$ bits. 
Again taking $m=(d+1)^2$ in \thmref{thm:mha} shows that the running time per arriving symbol is $\O{d^2\log d}$. 
\end{proofof}

We observe that the running time per arriving symbol can be improved to $\O{d^2}$ by creating a BZ sketch for the entire sliding window. However, this implementation uses $\O{d^2\log d}$ space instead. 

%


\section{Lower Bounds}
To prove \thmref{thm:lbone}, we first create a distribution between two strings, over which calculating the edit distance is equivalent to calculating the Hamming distance. 
We then show that any deterministic algorithm that approximates long length $d$-near-alignments under Hamming distance with high probability requires a certain amount of space through a simple counting argument. 
By Yao's Minimax Principle, any randomized algorithm with the same probability of success requires the same amount of space.

To prove \thmref{thm:lbone}, we define $X$ be the set of binary strings of length $n$ with $d$ many $1$'s.  
We pick $x$ independently and uniformly at random from $X$ and $y$ independently and uniformly at random from the set of binary strings of length $n$ with either $\HAM(x,y)=d$ or $\HAM(x,y)=d+1$. 
Define transformation $s(x)=x_1\textbf{1}^{d+1}x_2\textbf{1}^{d+1}\ldots\textbf{1}^{d+1}x_n\textbf{1}^{d+1}$. 
Thus, we pick $(S,T)\sim(s(x),s(y))$.

\newcommand{\clmcommute}{If $\ed(x,y)=d$, then there exist a sequence of $d$ insertions, deletions, or substitutions on $x$ to obtain $y$. Furthermore, we may perform the substitutions first, followed by the insertions, then the deletions.}
\begin{claim} \claimlab{clm:commute}
\clmcommute
\end{claim}
\begin{proof}
First, we fix a sequence of $d$ operations to obtain $y$ from $x$, and note that no character can be inserted and subsequently deleted, or else the edit distance between $x$ and $y$ would be less than $d$ by avoiding these operations. 
Similarly, any character which undergoes a substitution should not be involved in either an insertion or a deletion. 
Hence, any character is involved in at most one operation. 
But since a character is not affected by operations on other characters, we may first perform the substitutions, followed by the insertions, then the deletions.
\end{proof}

\begin{lemma}
$\ed(s(x),s(y))=\HAM(x,y)$
\end{lemma}
\begin{proof}
By \claimref{clm:commute}, we may perform the substitutions first, followed by the insertions, then the deletions to obtain $s(y)$ from $s(x)$. 
Let $s_1(x)$ be $s(x)$ following the sequence of substitutions.
Suppose there exists a position in $s_1(x)$ which does not equal the corresponding position in $s(y)$.
Then the position is zero in one of $s_1(x)$ or $s(y)$. 
However, the nearest zero in the other string is at least $d+1$ positions away, requiring at least $d+1$ additional operations. 
Since $\ed(s(x),s(y))\le\HAM(s(x),s(y))\le d+1$, then it follows that every single operation to obtain $s(y)$ from $s(x)$ must be a substitution, and so $\ed(s(x),s(y))\ge\HAM(s(x),s(y))$.
By construction, $\HAM(s(x),s(y))=\HAM(x,y)$ and the result follows.
\end{proof}

\newcommand{\lemmham}{Any algorithm $\mathcal{D}$ using less than $\frac{d\log n}{3}$ bits of memory cannot distinguish between $\HAM(x,y)=d$ and $\HAM(x,y)>d$ with probability at least $1-1/n$.}
\begin{lemma}\lemmlab{lemm:ham}
\lemmham
\end{lemma}
\begin{proof}
Note that $|X|=\binom{n}{d}$.
By Stirling's approximation, $|X|\ge\left(\frac{n}{d}\right)^d$. 
Since $d=o(\sqrt{n})$, then $|X|\ge\left(n\right)^{d/2}$.

If $\mathcal{D}$ uses less than $\frac{d\log n}{3}$ bits of memory, then $\mathcal{D}$ has at most $2^{\frac{d\log n}{3}}=n^{d/3}$ unique memory configurations. 
Since $|X|\ge\left(n\right)^{d/2}$, then there are at least $\frac{1}{2}(|X|-n^{d/3})\ge\frac{|X|}{4}$ pairs $x,x'$ such that $\mathcal{D}$ has the same configuration after reading $x$ and $x'$.
We show that $\mathcal{D}$ errs on a significant fraction of these pairs $x,x'$.

Let $\mathcal{I}$ be the positions where either $x$ or $x'$ take value $1$, so that $d+1\le|\mathcal{I}|\le 2d$.
Observe that if $\HAM(x,y)=d$, but $x$ and $y$ do not differ in any positions of $\mathcal{I}$, then $\HAM(x',y)>d$. 
Recall that $\mathcal{D}$ has the same configuration after reading $x$ and $x'$, but since $\HAM(x,y)=d$ and $\HAM(x',y)>d$, then the output of $\mathcal{D}$ is incorrect for either $\HAM(x,y)$ or $\HAM(x',y)$.

For each pair $(x,x')$, there are $\binom{n-|\mathcal{I}|}{d}\ge\binom{n-2d}{d}$ such $y$ with $\HAM(x,y)=d$, but $x$ and $y$ do not differ in any positions of $\mathcal{I}$. 
Hence, there are $\frac{|X|}{4}\binom{n-2d}{d}$ strings $s(x,y)$ for which $\mathcal{D}$ errs. 
We note that there is no overcounting because the output of $\mathcal{D}$ can be correct for at most one $\HAM(x_i,y)$ for all $x_i$ mapped to the same configuration. 
Recall that $y$ satisfies either $\HAM(x,y)=d$ or $\HAM(x,y)=d+1$ so that there are $|X|\left(\binom{n}{d}+\binom{n}{d+1}\right)$ pairs $(x,y)$ in total. 
Thus, the probability of error is at least
\[\frac{\frac{|X|}{4}\binom{n-2d}{d}}{|X|\left(\binom{n}{d}+\binom{n}{d+1}\right)}
=\frac{1}{4} \cdot \frac{\binom{n-2d}{d}}{\binom{n+1}{d+1}}\]
\[=\frac{(d+1)}{4}\frac{(n-3d+1)\ldots (n-2d)}{(n-d+1) \ldots  (n+1)}\]
Since $\frac{n-3d+1}{n-d+1}\le\frac{n-3d+i}{n-d+i}$ for all $i\ge1$, it follows that the probability of error is at least
\[\frac{d+1}{4(n+1)}\left(\frac{n-3d+1}{n-d+1}\right)^d
=\frac{d+1}{4n+4}\left(1-\frac{2d}{n-d+1}\right)^d\]
Then by Bernoulli's Inequality (which states that $(1+x)^r\ge 1+rx$ for $x\ge-1$ and $r\ge1$), the probability of error is at least
\[\frac{d+1}{4n+4}\left(1-\frac{2d^2}{n-d+1}\right)\ge\frac{1}{n}\]
Since $d=o(\sqrt{n})$, then for large $n$, it follows that $1-\frac{2d^2}{n-d+1}\ge\frac{1}{2}$. 
Hence, for $d>7$, the probability of error is at least $\frac{1}{n}$.

Therefore, $\Omega(d\log n)$ bits of memory are necessary to distinguish between $\HAM(x,y)=d$ and $\HAM(x,y)>d$ with probability at least $1-1/n$.
\end{proof}

Now, we use a simple trick to show that any sketch  providing a $(1+\eps)$-approximation to the length of the longest $d$-near-alignment under the edit distance with probability at least $1-1/n$ requires $\Omega(d\log n)$ space.

\begin{proofof}{\thmref{thm:lbone}}
Recall that $s(x)=x_1\textbf{1}^{d+1}x_2\textbf{1}^{d+1}\ldots\textbf{1}^{d+1}x_n\textbf{1}^{d+1}$. 
Define string $t(x)=\textbf{1}^{(d+1)n/2}x\textbf{1}^{(d+1)n/2}$ so that the longest $d$-near-alignment of $t(s(x))$ and $t(s(y))$ has length $2(d+1)n$ if $\ed(x,y)\le d$. 

On the other hand, if $\ed(x,y)>d$, then the longest $d$-near-alignment of $t(s(x))$ and $t(s(y))$ has length at most $(d+1)n$. 
Thus, a $(1+\eps)$-approximation to the length of the longest $d$-near-alignment of $t(s(x))$ and $t(s(y))$ differentiates whether $\HAM(x,y)\le d$ or $\HAM(x,y)>d$.
Since $t(s(x))$ has length $2(d+1)n$, any sketch which achieves this requires $\Omega(d\log(n/d))$ bits. 
Because $d=o(\sqrt{n})$, then the result follows.
\end{proofof}

We now turn our attention to \thmref{thm:lbtwo}, which states that any algorithm computing a $(1+\eps)$-multiplicative or $E$-additive approximation of the length of the longest $d$-near alignment under the edit distance and outputs the necessary edit operations requires $\Omega(d\log n)$ bits, even in the simultaneous streaming model. 
Furthermore, simply determining the length of the longest $d$-near alignment also requires $\Omega(d\log n)$ bits.

\begin{proofof}{\thmref{thm:lbtwo}}
We first prove that any algorithm that computes a $(1+\eps)$-multiplicative approximation of the length of the longest $d$-near alignment under the edit distance requires $\Omega(d\log n)$ bits using a reduction from the corresponding problem from communication complexity. Namely, in the communication complexity model, Alice receives  the first half of both $S$ and $T$, and Bob receives the second half of $S$ and $T$; their goal is to find the longest $d$-near-alignment between $S$ and $T$.  
Now, suppose $S\left[1,\frac{n}{2}\right]$ and $T\left[1,\frac{n}{2}\right]$ have edit distance $d$, and none of the edit operations occur within the first $\left(1-\frac{1}{1+\eps}\right)n$ positions of $S$ and $T$. 
Thus, Alice must communicate the locations of all edit operations (i.e., $\Omega(d\log n)$ bits.), as any one of these locations could be the beginning of the longest $d$-near-alignment. 

We observe that an algorithm that computes a $E$-additive approximation of the length of the longest $d$-near alignment under the edit distance and outputs the necessary edit operations also forces Alice to communicate the locations of the $d$ most recent edit operations, provided that $S\left[1,\frac{n}{2}\right]$ and $T\left[1,\frac{n}{2}\right]$ have edit distance $d$ and none of the edit operations occur in the first $E$ locations of Alice's input.

Finally, if Alice and Bob must output the length of the longest $d$-near-alignment, and $S\left[1,\frac{n}{2}\right]$ and $T\left[1,\frac{n}{2}\right]$ have edit distance $d$, then Alice must output the locations of the $d$ most recent edit operations. 
\end{proofof}
\def\shortbib{0}
\bibliographystyle{plain}
\bibliography{references}

\begin{thebibliography}{10}

\bibitem{Alon1996}
Noga Alon, Yossi Matias, and Mario Szegedy.
\newblock The space complexity of approximating the frequency moments.
\newblock In {\em Proceedings of the twenty-eighth annual ACM symposium on
  Theory of computing}, pages 20--29. ACM, 1996.

\bibitem{AluruAT15}
Srinivas Aluru, Alberto Apostolico, and Sharma~V. Thankachan.
\newblock Efficient alignment free sequence comparison with bounded mismatches.
\newblock In {\em Research in Computational Molecular Biology - 19th Annual
  International Conference, {RECOMB}, Proceedings}, pages 1--12, 2015.

\bibitem{AndoniGMP13}
Alexandr Andoni, Assaf Goldberger, Andrew McGregor, and Ely Porat.
\newblock Homomorphic fingerprints under misalignments: sketching edit and
  shift distances.
\newblock In {\em Proceedings of the Forty-Seventh Annual {ACM} on Symposium on
  Theory of Computing, {STOC}}, pages 931--940, 2013.

\bibitem{ApostolicoG97}
Alberto Apostolico and Zvi Galil, editors.
\newblock {\em Pattern Matching Algorithms}.
\newblock Oxford University Press, Oxford, UK, 1997.

\bibitem{BackursI15}
Arturs Backurs and Piotr Indyk.
\newblock Edit distance cannot be computed in strongly subquadratic time
  (unless {SETH} is false).
\newblock In {\em Proceedings of the Forty-Seventh Annual {ACM} on Symposium on
  Theory of Computing, {STOC}}, pages 51--58, 2015.

\bibitem{BelazzouguiZ16}
Djamal Belazzougui and Qin Zhang.
\newblock Edit distance: Sketching, streaming and document exchange.
\newblock In {\em 57th Annual Symposium on Foundations of Computer Science,
  {FOCS}}, pages 51--60, 2016.

\bibitem{BerenbrinkEMS14}
Petra Berenbrink, Funda Erg{\"{u}}n, Frederik Mallmann{-}Trenn, and Erfan
  {Sadeqi Azer}.
\newblock Palindrome recognition in the streaming model.
\newblock In {\em 31st International Symposium on Theoretical Aspects of
  Computer Science {(STACS})}, pages 149--161, 2014.

\bibitem{Chakrabarti2015b}
Amit Chakrabarti.
\newblock Data stream algorithms.
\newblock {\em Computer Science}, 49:149, 2015.

\bibitem{ChakrabortyGK16a}
Diptarka Chakraborty, Elazar Goldenberg, and Michal Kouck{\'{y}}.
\newblock Streaming algorithms for computing edit distance without exploiting
  suffix trees.
\newblock {\em CoRR}, abs/1607.03718, 2016.

\bibitem{ChakrabortyGK16}
Diptarka Chakraborty, Elazar Goldenberg, and Michal Kouck{\'{y}}.
\newblock Streaming algorithms for embedding and computing edit distance in the
  low distance regime.
\newblock In {\em Proceedings of the 48th Annual {ACM} {SIGACT} Symposium on
  Theory of Computing, {STOC}}, pages 712--725, 2016.

\bibitem{CliffordFPSS16}
Rapha{\"{e}}l Clifford, Allyx Fontaine, Ely Porat, Benjamin Sach, and
  Tatiana~A. Starikovskaya.
\newblock The \emph{k}-mismatch problem revisited.
\newblock In {\em Proceedings of the Twenty-Seventh Annual {ACM-SIAM} Symposium
  on Discrete Algorithms, {SODA}}, pages 2039--2052, 2016.

\bibitem{ErgunGSZ17}
Funda Erg{\"{u}}n, Elena Grigorescu, Erfan {Sadeqi Azer}, and Samson Zhou.
\newblock Streaming periodicity with mismatches.
\newblock In {\em Approximation, Randomization, and Combinatorial Optimization.
  Algorithms and Techniques - 21st International Workshop, {RANDOM} (to
  appear)}, 2017.

\bibitem{FlouriGKU15}
Tom{\'{a}}s Flouri, Emanuele Giaquinta, Kassian Kobert, and Esko Ukkonen.
\newblock Longest common substrings with k mismatches.
\newblock {\em Inf. Process. Lett.}, 115(6-8):643--647, 2015.

\bibitem{GawrychowskiMSU16}
Pawel Gawrychowski, Oleg Merkurev, Arseny~M. Shur, and Przemyslaw Uznanski.
\newblock Tight tradeoffs for real-time approximation of longest palindromes in
  streams.
\newblock In {\em 27th Annual Symposium on Combinatorial Pattern Matching,
  {CPM}}, pages 18:1--18:13, 2016.

\bibitem{Gilbert2010}
Anna~C. Gilbert and Piotr Indyk.
\newblock Sparse recovery using sparse matrices.
\newblock {\em Proceedings of the {IEEE}}, 98(6):937--947, 2010.

\bibitem{GrigorescuSZ17}
Elena Grigorescu, Erfan {Sadeqi Azer}, and Samson Zhou.
\newblock Streaming for aibohphobes: Longest palindrome with mismatches.
\newblock {\em CoRR}, abs/1705.01887, 2017.

\bibitem{Hirschberg75}
Daniel~S. Hirschberg.
\newblock A linear space algorithm for computing maximal common subsequences.
\newblock {\em Commun. {ACM}}, 18(6):341--343, 1975.

\bibitem{Hui92}
Lucas Chi~Kwong Hui.
\newblock Color set size problem with application to string matching.
\newblock In {\em Combinatorial Pattern Matching, Third Annual Symposium, {CPM}
  Proceedings}, pages 230--243, 1992.

\bibitem{kleinberg2006algorithm}
Jon Kleinberg and Eva Tardos.
\newblock Algorithm design, 2006.

\bibitem{KociumakaSV14}
Tomasz Kociumaka, Tatiana~A. Starikovskaya, and Hjalte~Wedel Vildh{\o}j.
\newblock Sublinear space algorithms for the longest common substring problem.
\newblock In {\em Algorithms - {ESA}. Proceedings}, pages 605--617, 2014.

\bibitem{LeimeisterM14}
Chris{-}Andre Leimeister and Burkhard Morgenstern.
\newblock kmacs: the \emph{k}-mismatch average common substring approach to
  alignment-free sequence comparison.
\newblock {\em Bioinformatics}, 30(14):2000--2008, 2014.

\bibitem{Liben-NowellVZ06}
David Liben{-}Nowell, Erik Vee, and An~Zhu.
\newblock Finding longest increasing and common subsequences in streaming data.
\newblock {\em J. Comb. Optim.}, 11(2):155--175, 2006.

\bibitem{Muthukrishnan2004}
S~Muthukrishnan.
\newblock Data stream algorithms, 2004.

\bibitem{NeedlemanW70}
Saul~B. Needleman and Christian~D. Wunsch.
\newblock A general method applicable to the search for similarities in the
  amino acid sequence of two proteins.
\newblock {\em Journal of Molecular Biology}, 48(3):443--53, 1970.

\bibitem{PoratP09}
Benny Porat and Ely Porat.
\newblock Exact and approximate pattern matching in the streaming model.
\newblock In {\em 50th Annual {IEEE} Symposium on Foundations of Computer
  Science, {FOCS}}, pages 315--323, 2009.

\bibitem{Sankoff72}
David Sankoff.
\newblock Matchings sequences under deletion/insertion constraints.
\newblock {\em Proceedings of the National Academy of Sciences}, 69(1):1--4,
  1972.

\bibitem{Sellers74}
Peter~H. Sellers.
\newblock On the theory and computation of evolutionary distances.
\newblock {\em SIAM Journal on Applied Mathematics}, 26(4):787--793, 1974.

\bibitem{Starikovskaya16}
Tatiana~A. Starikovskaya.
\newblock Longest common substring with approximately k mismatches.
\newblock In {\em 27th Annual Symposium on Combinatorial Pattern Matching,
  {CPM}}, pages 21:1--21:11, 2016.

\bibitem{SunW07}
Xiaoming Sun and David~P. Woodruff.
\newblock The communication and streaming complexity of computing the longest
  common and increasing subsequences.
\newblock In {\em Proceedings of the Eighteenth Annual {ACM-SIAM} Symposium on
  Discrete Algorithms, {SODA}}, pages 336--345, 2007.

\bibitem{Ukkonen85}
Esko Ukkonen.
\newblock Finding approximate patterns in strings.
\newblock {\em J. Algorithms}, 6(1):132--137, 1985.

\bibitem{Vintsiuk68}
Taras Vintsiuk.
\newblock Speech discrimination by dynamic programming.
\newblock {\em Kibernetika}, 4(1):81--88, 1968.

\bibitem{WagnerF74}
Robert~A. Wagner and Michael~J. Fischer.
\newblock The string-to-string correction problem.
\newblock {\em Journal of the ACM}, 21:168--178, 1974.

\bibitem{Weiner73}
Peter Weiner.
\newblock Linear pattern matching algorithms.
\newblock In {\em 14th Annual Symposium on Switching and Automata Theory,{SWAT
  (FOCS)}}, pages 1--11, 1973.

\end{thebibliography}
\end{document}